\documentclass[10pt]{article}

\usepackage{graphicx}
\usepackage{amsmath, amsthm, amssymb}
\usepackage[boxruled]{algorithm2e}
\usepackage{subfigure}
\usepackage{comment}
\usepackage{url}
\usepackage{xcolor}
\usepackage{tikz}

\newcommand{\E}{\mathbb{E}}

\newcommand{\ket}[1]{| #1 \rangle}

\newcommand{\be}{\begin{equation}}
\newcommand{\ee}{\end{equation}}
\newcommand{\bea}{\begin{eqnarray}}
\newcommand{\eea}{\end{eqnarray}}
\newcommand{\bes}{\begin{equation*}}
\newcommand{\ees}{\end{equation*}}
\newcommand{\beas}{\begin{eqnarray*}}
\newcommand{\eeas}{\end{eqnarray*}}

\newtheorem{thm}{Theorem}[section]
\newtheorem*{thm*}{Theorem}
\newtheorem{cor}[thm]{Corollary}

\newtheorem{lem}[thm]{Lemma}
\newtheorem*{lem*}{Lemma}
\newtheorem{prop}[thm]{Proposition}

\begin{document}

\date{\today}

\title{Quantum search with advice}

\author{Ashley Montanaro\footnote{montanar@cs.bris.ac.uk}\\ \\ {\small Department of Computer Science, University of Bristol,}\\{\small Woodland Road, Bristol, BS8 1UB, UK.} }

\maketitle

\begin{abstract}
We consider the problem of search of an unstructured list for a marked element $x$, when one is given advice as to where $x$ might be located, in the form of a probability distribution. The goal is to minimise the expected number of queries to the list made to find $x$, with respect to this distribution. We present a quantum algorithm which solves this problem using an optimal number of queries, up to a constant factor. For some distributions on the input, such as certain power law distributions, the algorithm can achieve exponential speed-ups over the best possible classical algorithm. We also give an efficient quantum algorithm for a variant of this task where the distribution is not known in advance, but must be queried at an additional cost. The algorithms are based on the use of Grover's quantum search algorithm and amplitude amplification as subroutines.
\end{abstract}


\section{Introduction}

Grover's algorithm for search of an unstructured list is one of the greatest successes of the nascent field of quantum computation \cite{grover97}. The algorithm operates in the {\em black box} model: given access to a function $f:\{1,\dots,n\} \rightarrow \{0,1\}$, where $f$ is promised to take the value 1 on precisely one input $x$, it finds $x$ with certainty using $O(\sqrt{n})$ queries to $f$, whereas any classical algorithm requires $\Omega(n)$ queries to perform the same task. However, it is rarely necessary to search the type of databases that we encounter in real life in a completely unstructured fashion. Instead, there is often some prior information about the location of the sought (``marked'') item $x$, which can be used to guide the search. We can formalise this intuition by considering a search problem where the searcher is given access to a probability distribution, which hints where the marked item is likely to be. 

\begin{figure}[h]
\noindent \framebox{
\begin{minipage}{14.6cm}
\noindent {\bf Problem:} {\sc Search with Advice}\\
\noindent {\bf Input:} A function $f:\{1,\dots,n\} \rightarrow \{0,1\}$ that takes the value 1 on precisely one input $x$, and an ``advice'' probability distribution $\mu = (p_y)$, $y \in \{1,\dots,n\}$, where $p_y$ is the probability that $f(y)=1$.\\
\noindent {\bf Output:} The marked element $x$.
\end{minipage}
}
\end{figure}

It is clear that knowledge of $\mu$ can enable a classical algorithm to achieve a significant reduction in the average number of queries to $f$ (with respect to $\mu$) required to find the marked element $x$. This paper is concerned with the development of {\em quantum} algorithms for the {\sc Search with Advice} problem which also use $\mu$, and which obtain significant speed-ups over any classical algorithm.

We distinguish two models for the complexity of this problem. In the first model -- the {\em known} model -- $\mu$ is known completely beforehand, and can be used to help design an algorithm to find the marked element. The complexity of the problem is given by the minimum expected number of queries to $f$ required to find $x$, under the distribution $\mu$. In the second model -- the {\em unknown} model -- $\mu$ is not known before the algorithm starts, but the algorithm is also given access to a black box which outputs samples from $\mu$, at unit cost. In the case of quantum algorithms, the black box outputs a coherent superposition corresponding to $\mu$ (a ``quantum sample'').

In both cases, note that we are interested in the {\em average} number of queries with respect to $\mu$ required to find the marked element, rather than the worst-case number of queries. Previous work has shown that, if one considers the worst-case number of queries to the input required to compute any total function, there can only be at most a polynomial separation between quantum and classical computation \cite{beals01}. Considering the average number of queries required (over the input) allows one to sidestep these results and hope to obtain {\em exponential} speed-ups.

Indeed, previous work of Ambainis and de Wolf \cite{ambainis01} has shown that quantum algorithms {\em can} achieve exponential (or even super-exponential) reductions in average-case query complexity over classical algorithms. The model that these authors considered was that of computing a particular boolean function $f:\{0,1\}^n \rightarrow \{0,1\}$, with a particular (known) distribution on the inputs. Among other results, they exhibited a (function, distribution) pair with a super-exponential separation between quantum and classical query complexity, and even gave a function whose quantum and classical query complexity were exponentially separated under the {\em uniform} distribution.


\subsection{New results}

The main results of this paper are as follows. First, in the known model, we give a quantum algorithm for {\sc Search with Advice} which is optimal up to constant factors. Assuming without loss of generality that the probability distribution $\mu = (p_x)$ is given in non-increasing order, the algorithm uses an expected number of queries to $f$ which is of the order of
\[ \sum_{x=1}^n p_x \sqrt{x}, \]
which should be compared with the optimal classical expected number of queries,
\[ \sum_{x=1}^n p_x\,x. \]
For certain probability distributions, this represents an exponential (or even super-exponential) improvement in the expected number of queries used. The quantum algorithm is based on the use of an exact variant of Grover's algorithm \cite{grover97,hoyer00,brassard02} as a subroutine. Known lower bounds on the query complexity of quantum search are used to show that this algorithm is optimal, up to constant factors, for any probability distribution $\mu$.

In the unknown model, we give a quantum algorithm that uses a expected number of queries of the order of
\[ \left( \sum_{x,p_x > 1/n} \sqrt{p_x} \right) + \sqrt{n} \left( \sum_{x,p_x \le 1/n} p_x \right). \]
Again, this algorithm is sometimes significantly more efficient than the best possible classical algorithm. The algorithm is based on the amplitude amplification algorithms proposed by Boyer et al \cite{boyer98} and Brassard et al \cite{brassard02}; the main difference being that after performing a certain number of iterations of amplitude amplification, it reverts to exact Grover search. This can considerably improve the average query complexity.

These results in the two different models are applied to the natural class of power law distributions $p_x \propto x^{-k}$, for some constant $k > 0$. We will see that for certain values of $k$, quantum algorithms deliver very significant reductions in the average number of queries used. In particular, when $-2 < k < -3/2$, a super-exponential separation between quantum and classical computation is obtained in the known model ($O(1)$ vs.\ $\Omega(n^{k+2})$). The results for power law distributions are summarised in Figure \ref{fig:graph} (with details given in Propositions \ref{prop:powerlaw} and \ref{prop:powerlaw2} below).

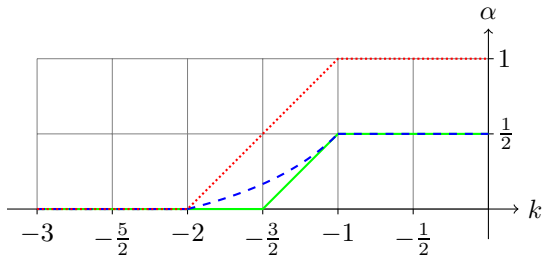
\begin{figure}
\begin{center}
\begin{tikzpicture}[domain=-3:0,xscale=2,yscale=2]
\draw[ultra thin,color=gray,step=0.5] (-3,0) grid (0,1);
\draw[->] (-3.2,0) -- (0.2,0) node[right] {$k$};
\draw[->] (0,-0.2) -- (0,1.2) node[above] {$\alpha$};
\draw (-0.5,0) -- (-0.5,-1pt) node [anchor=north] {$-\frac{1}{2}$};
\draw (-1,0) -- (-1,-1pt) node [anchor=north] {$-1$};
\draw (-1.5,0) -- (-1.5,-1pt) node [anchor=north] {$-\frac{3}{2}$};
\draw (-2,0) -- (-2,-1pt) node [anchor=north] {$-2$};
\draw (-2.5,0) -- (-2.5,-1pt) node [anchor=north] {$-\frac{5}{2}$};
\draw (-3,0) -- (-3,-1pt) node [anchor=north] {$-3$};
\draw (1pt,0.5) -- (0,0.5) node[anchor=west] {$\frac{1}{2}$};
\draw (1pt,1) -- (0,1) node[anchor=west] {$1$};

\draw[color=green,thick] (-3,0) -- (-1.5,0) -- (-1,0.5) -- (0,0.5);
\draw[color=blue,dashed,thick] (-3,0) -- (-2,0);
\draw[color=blue,dashed,thick,domain=-2:-1] plot (\x,-0.5-1/\x);
\draw[color=blue,dashed,thick] (-1,0.5) -- (0,0.5);
\draw[color=red,densely dotted,thick] (-3,0) -- (-2,0) -- (-1,1) -- (0,1);
\end{tikzpicture}
\label{fig:graph}
\caption{Query complexity of {\sc Search with Advice} in different models, for power law distributions $p_x \propto x^{-k}$. For each $k$, the query complexity of the algorithms given in this paper is $\Theta(n^{\alpha})$ for some $\alpha$ (ignoring log factors); the graph plots the exponent $\alpha$ against $k$. Dotted red line: best classical algorithm; solid green line: quantum, known probability distribution; dashed blue line: quantum, unknown probability distribution. }
\end{center}
\end{figure}

The paper is organised as follows. After defining the models and notation used, Section \ref{sec:known} contains the results on the known model, while Section \ref{sec:unknown} studies the unknown model. Various proofs are deferred to an appendix.


\subsection{Models and notation}

In this section, we set up concepts and notation that will be used throughout the paper. We assume familiarity with quantum computation \cite{nielsen00}, and in particular the concept of query complexity \cite{buhrman02}: the number of queries to the input which a classical or quantum algorithm requires to compute some function. Let $[n]$ denote the integers $\{1,\dots,n\}$, and consider an oracle function $f:[n] \rightarrow \{0,1\}$ which is promised to take the value 1 on precisely one input $x \in [n]$ (that is, $f(y) = \delta_{xy}$). We say that $x$ is the {\em marked} element. Also consider a quantum or classical algorithm $\mathcal{A}$ which, given access to $f$, attempts to output $x$. We say that $\mathcal{A}$ is a {\em valid} algorithm if, for any $f$ satisfying the above constraint, $\mathcal{A}$ outputs $x$ with certainty. Let $\mathcal{D}$ denote the set of valid deterministic classical algorithms, and let $\mathcal{Q}$ denote the set of valid quantum algorithms.

Let $\mu = (p_x)$ be a distribution on $[n]$ giving the probability for the marked element to be found at each location. We will be concerned with understanding the average number of queries to the input required to find $x$. This will only depend on $\mu$, and will hence be termed the (average-case) query complexity of $\mu$. The general model of average-case query complexity used here will be similar to that in \cite{ambainis01}, with some minor differences. We distinguish two models for the {\sc Search with Advice} problem: a {\em known} model and an {\em unknown} model. In the known model, the probability distribution $\mu$ is known beforehand, and can be used to design the algorithm. In the unknown model, $\mu$ is not known, and the algorithm must query an oracle to gain information about $\mu$.

We first define the known model. Let $\mathcal{A}$ be a valid algorithm, and let $T_{\mathcal{A}}(x)$ denote the expected number of queries to $f$ used by $\mathcal{A}$, when $x$ is the marked element. Note that, in order for this model to be interesting in the case where $\mathcal{A}$ is a quantum algorithm, intermediate measurements during the search process are allowed; otherwise, $\mathcal{A}$ would always use the same number of queries. Further, let $T_{\mathcal{A}}(\mu)$ be the expected number of queries to $f$ used by $\mathcal{A}$, where the expectation is taken over both the distribution $\mu$ and (potentially) $\mathcal{A}$'s internal randomness. That is,
\[ T_{\mathcal{A}}(\mu) = \sum_{x=1}^n p_x T_{\mathcal{A}}(x). \]
Finally, we define the main quantities of interest, the deterministic and quantum (respectively) average-case query complexities of $\mu$.
\beas
D(\mu) &=& \min_{\mathcal{A} \in \mathcal{D}} T_{\mathcal{A}}(\mu),\\
Q(\mu) &=& \min_{\mathcal{A} \in \mathcal{Q}} T_{\mathcal{A}}(\mu).
\eeas
The restriction to algorithms that succeed with certainty makes this a zero-error ({\em Las Vegas}) notion of average-case query complexity. It is common to consider an alternative {\em Monte Carlo} model of query complexity where $\mathcal{A}$ is allowed to err with some constant probability (e.g.\ 1/3). Note that this would not change the model significantly in the case of the current problem: given a (classical or quantum) Monte Carlo search algorithm that uses $t$ queries and outputs $x$ with probability $p$, one can produce an algorithm that succeeds with certainty and uses an expected number of queries of at most $(t+1)/p$ \cite[Exercise 1.3]{motwani95}.

We now turn to the unknown model. In this scenario, as well as querying $f$, we allow $\mathcal{A}$ to sample from $\mu$ using an oracle. In the case of quantum algorithms, we allow the preparation of {\em quantum} samples; that is, in the quantum case we define an oracle $O_\mu$, which performs the mapping
\[ O_\mu \ket{0} = \ket{\mu} := \sum_{x=1}^n \sqrt{p_x} \ket{x}. \]
The algorithm is also given access to the inverse operation, $O_\mu^{-1}$. We define $T^*_{\mathcal{A}}(\mu)$ as the expected total number of queries to $f$, $O_\mu$ and $O_\mu^{-1}$ used by $\mathcal{A}$ (a query to each oracle being counted as unit cost).
%
%
The oracle $O_\mu$ may appear somewhat unrealistic. However, it can be implemented if one has the ability to sum the distribution $\mu$ over arbitrary ranges \cite{grover02}. That is, given an efficient means of computing $\sum_{x=a}^b p_x$ for arbitrary $a$, $b$, one can implement $O_\mu$ efficiently.

Finally, we will make use of an exact variant of Grover's quantum search algorithm throughout this paper.
\begin{thm}[Grover \cite{grover97}, H\o yer \cite{hoyer00}, Brassard et al \cite{brassard02}]
\label{thm:grover}
Given an unstructured list of $n$ elements that contains a unique marked element, there is a quantum algorithm that finds the marked element with certainty using $\lceil\frac{\pi}{4}\sqrt{n}\rceil$ queries to the list. If the list is promised to contain either one or zero marked elements, the marked element can be found (or ``no marked element present'' returned) with certainty using one extra query.
\end{thm}


\section{Search with a known probability distribution}
\label{sec:known}

In this section, we will assume that $p_x$ is non-increasing with $x$ (so the most likely place for the marked element to be is at the start of the list, etc.). With this assumption, the optimal classical algorithm to find $x$ is simply to query $f(1)$ through $f(n)$ in turn, so the classical average-case query complexity can be written down as
\be
\label{eqn:deterministic}
D(\mu) = \sum_{x=1}^n p_x\,x.
\ee
Note that, classically, the algorithm obtains no benefit from the use of randomness. When $\mu$ is the uniform distribution, corresponding to having no information about the location of the marked item, Grover's algorithm (Theorem \ref{thm:grover}) achieves a quadratic reduction in average-case query complexity. However, na\"ive use of this algorithm does not give an advantage in the average-case setting in general. In the next section, we give a quantum algorithm which {\em does} significantly improve on the trivial classical algorithm above.


\subsection{Geometric search algorithm}

We now give an algorithm for the {\sc Search with Advice} problem, which will turn out to be asymptotically optimal. The quantum component of this algorithm is in fact simply Grover search (Theorem \ref{thm:grover}). Informally, the algorithm consists of splitting the input into blocks which increase in size geometrically (hence its name) and performing Grover search on each block. Interestingly, the algorithm does not need to know the precise advice probability distribution to achieve its near-optimal query complexity: it suffices to be able to sort the probabilities in non-increasing order. The algorithm is parametrised by a constant $k$, which gives the ratio of the geometric progression. We optimise $k$ below; however, changing $k$ only affects the query complexity by a constant factor.

\begin{algorithm}
\KwIn{Advice distribution $\mu = (p_x)$ in non-increasing order; function $f:[n] \rightarrow \{0,1\}$ such that $f$ takes the value 1 on precisely one input $x$; real $k>1$}
\KwOut{The marked element $x$}
$start \leftarrow 1$\;
$end \leftarrow$ 1\;
$step \leftarrow$ 0\;
\While{start $\le n$}{
  perform exact Grover search for one or zero marked elements on subset $\{start,\dots,end\}$\;
  \If{marked element found}{
    \Return{marked element}\;
  }
  $step \leftarrow step + 1$\;
  $start \leftarrow end + 1$\;
  $end \leftarrow \min(start + \lfloor k^{step} \rfloor - 1,n)$\;
}
\Return{error}\;
\caption{Geometric quantum search\vspace{\baselineskip}}
\label{alg:geometric}
\end{algorithm}

\begin{prop}
\label{prop:average}
The average number of queries used by Algorithm \ref{alg:geometric}, choosing $k=e \approx 2.718$, on an advice distribution $\mu = (p_x)$ is upper bounded by
\[ \pi e \sum_{x=1}^n p_x \sqrt{x}. \]
\end{prop}

\begin{proof}
In the $m$'th iteration of the loop, the (at most) $\lfloor k^m \rfloor$ elements contained in the range
\be
\label{eqn:range}
R_m = \{1 + \sum_{i=0}^{m-1} \lfloor k^i \rfloor, \dots , \min(\lfloor k^m \rfloor - \sum_{i=0}^{m-1} \lfloor k^i \rfloor,n)\}
\ee
will be searched. By Theorem \ref{thm:grover}, the Grover search step in this iteration uses $\left\lceil\frac{\pi}{4}\sqrt{\lfloor k^m \rfloor}\right\rceil + 1$ queries. Then, for any marked element $x \in R_m$, a total of at most
\[ \sum_{s=0}^m \left( \left\lceil\frac{\pi}{4}\sqrt{\lfloor k^s \rfloor}\right\rceil + 1 \right) \le 2(m+1) + \frac{\pi}{4} \sum_{s=0}^m k^{s/2} \]
queries will be used by Algorithm \ref{alg:geometric} to find $x$. It is clear from (\ref{eqn:range}) that, for any $x \in R_m$, $m \le \log_k x + 1$. The average-case query complexity is therefore upper bounded by
\[ \sum_{x=1}^n p_x \left( 2\log_k x + 4 + \frac{\pi}{4} \sum_{s=0}^{\lfloor \log_k x + 1 \rfloor} k^{s/2} \right), \]
and, estimating the inner sum by an integral, we obtain an upper bound of
\[ \sum_{x=1}^n p_x \left( 2 \log_k x + 4 + \frac{\pi}{4} \int_0^{\log_k x + 2} k^{s/2}\,ds \right) = 2 \sum_{x=1}^n p_x \left( \log_k x + 2 + \frac{\pi k}{4 \ln k} \sqrt{x} \right). \]
Picking $k=e$, and noting that $\ln x + 2 \le \frac{\pi e}{4} \sqrt{x}$ for all $x > 0$, completes the proof.
\end{proof}


\subsection{Optimality of the geometric search algorithm}

We now show that Algorithm \ref{alg:geometric} is in fact optimal, up to a constant factor. This result will rely on the following known exact bound on the query complexity of quantum search.

\begin{thm}[Grover \cite{grover97}, Zalka \cite{zalka99}]
\label{thm:lower}
Let $f:[n] \rightarrow \{0,1\}$ be a function that takes the value 1 on precisely one input $x$, and let $\mathcal{A}$ be a quantum search algorithm that uses $T$ queries to $f$ and outputs $x$ with probability at least $p$, for all $x$. Then
\[ T \ge \left\lceil \frac{\arcsin \sqrt{p}}{2 \arcsin(1/\sqrt{n})} - \frac{1}{2} \right\rceil, \]
and this number of queries is achieved by Grover's algorithm.
\end{thm}

As stated, this bound involves worst-case query complexity (that is, the largest possible number of queries used by $\mathcal{A}$, on the worst possible input). In our setting, we will need to lower bound the {\em expected} number of queries used by $\mathcal{A}$ on the worst possible input. This can be done with the following proposition.

\begin{prop}
\label{prop:lasvegaslower}
Let $\mathcal{A}$ be a valid quantum search algorithm such that $T_{\mathcal{A}}(x) \le T$ for all $x$, for some $T_0$. Then
\[ T \ge \frac{0.206}{\arcsin{1/\sqrt{n}}} - 0.316 \ge 0.206\sqrt{n} - 1.\]
\end{prop}

\begin{proof}
Let $t_{\mathcal{A}}(x)$ be the random variable giving the number of queries used by $\mathcal{A}$ on input $x$. Thus $T_{\mathcal{A}}(x) = \E\,t_{\mathcal{A}}(x)$, where the expectation is taken over $\mathcal{A}$'s internal randomness. By Markov's inequality, for all $x$ and all $0 < p < 1$,
\[ \Pr[t_{\mathcal{A}}(x) \ge T_{\mathcal{A}}(x)/(1-p)] \le (1-p). \]
Thus a quantum search algorithm $\mathcal{A}$ that uses an expected number of at most $T$ queries on all $x$ gives a bounded-error quantum search algorithm that uses at most $T/(1-p)$ queries on all $x$ and succeeds with probability at least $p$: just run $\mathcal{A}$ until it has used $T/(1-p)$ queries, and if it has not output $x$, output a random integer between 1 and $n$. By Markov's inequality, this will succeed with probability at least $p$.

So, by Theorem \ref{thm:lower}, we have that for any $0 < p < 1$
\[ T \ge (1-p)\left(\frac{\arcsin \sqrt{p}}{2 \arcsin(1/\sqrt{n})} - \frac{1}{2}\right). \]
Performing numerical maximisation of the right-hand side over $p$, one finds that for large $n$ the maximum is achieved at $p \approx 0.369$, which proves the proposition.
\end{proof}

Note that it is known that one can indeed achieve an expected query complexity that is somewhat less than the usual worst-case query complexity guaranteed by Grover's algorithm \cite{boyer98,zalka99}. By stopping and restarting Grover search, it is possible to find the marked element $x$ using approximately $0.690\sqrt{n}$ expected queries on all $x$, whereas straightforward use of Grover's algorithm guarantees approximately $0.785\sqrt{n}$ queries.

We are now ready to prove that Algorithm \ref{alg:geometric} is asymptotically optimal.

\begin{prop}
Let $\mu = (p_x)$, $x \in [n]$ be an arbitrary probability distribution. Then
\[ Q(\mu) \ge 0.206 \sum_{x=1}^n p_x \sqrt{x} - 1. \]
\end{prop}

\begin{proof}
Let $\mathcal{A}$ be a valid quantum search algorithm and assume that $\mu$ is non-increasing. We aim to lower bound $T_{\mathcal{A}}(\mu) = \sum_{x=1}^n p_x T_{\mathcal{A}}(x)$. By Proposition \ref{prop:lasvegaslower}, there must exist a $y$ such that $T_{\mathcal{A}}(y) \ge 0.206\sqrt{n} - 1$. Similarly, there must exist $y' \neq y$ such that $T_{\mathcal{A}}(y') \ge 0.206\sqrt{n-1} - 1$ (or $\mathcal{A}$ would be able to find a marked element in the set of all elements not equal to $y$, using a number of queries that violates Proposition \ref{prop:lasvegaslower}). Iterating this argument, we see that for each $k$ such that $1 \le k \le n$ there exists an $x$ such that $T_{\mathcal{A}}(x) \ge 0.206\sqrt{k} - 1$. By a rearrangement inequality, this implies that
\[ T_{\mathcal{A}}(\mu) \ge \sum_{x=1}^n p_x (0.206\sqrt{x} - 1) \]
and proves the proposition.
\end{proof}


\subsection{Power law distributions}

We now apply Algorithm \ref{alg:geometric} to a natural class of probability distributions: power law distributions. We will see that significant speed-ups can be obtained over any possible classical algorithm.

\newcounter{powerlaw}\setcounter{powerlaw}{\value{thm}}
\newcounter{powerlawsec}\setcounter{powerlawsec}{\value{section}}
\begin{prop}
\label{prop:powerlaw}
Let $\mu = (p_x)$, $x \in [n]$ be a probability distribution where $p_x \propto x^k$ for some constant $k<0$. Then
\begin{equation*}
\begin{array}{cc}
D(\mu) = \left\{ 
\begin{array}{ll}
 \Theta(n) &\,\,[-1 < k < 0] \\
 \Theta(n/\log n) &\,\,[k=-1] \\
 \Theta(n^{k+2}) &\,\,[-2 < k < -1] \\
 \Theta(\log n) &\,\,[k=-2] \\
 \Theta(1) &\,\,[k<-2]
\end{array}
\right.
&, \mbox{ and }
Q(\mu) = \left\{ 
\begin{array}{ll}
  \Theta(\sqrt{n}) &\,\,[-1 < k < 0] \\
  \Theta(\sqrt{n}/\log n) &\,\,[k=-1] \\
  \Theta(n^{k+3/2}) &\,\,[-3/2 < k < -1] \\
  \Theta(\log n) &\,\,[k=-3/2] \\
  \Theta(1) &\,\,[k<-3/2]
  \end{array}
\right.
\end{array}
\end{equation*}
\end{prop}

\begin{proof}
Deferred to Appendix.
\end{proof}

\begin{cor}
There exists a probability distribution $\mu$ such that $D(\mu) = \Omega(n^{1/2-\epsilon})$ for arbitrary $\epsilon > 0$, but $Q(\mu) = O(1)$.
\end{cor}

\begin{proof}
Take $k = -3/2 - \epsilon$ in Proposition \ref{prop:powerlaw}. (Indeed, any $k \in (-2,-3/2)$ gives a super-exponential separation between $D(\mu)$ and $Q(\mu)$.)
\end{proof}


\section{Unknown probability distribution}
\label{sec:unknown}

In this section we switch to a different model, where the algorithm does not know the advice distribution $\mu$ in advance, but must use an oracle to obtain information about this distribution. We begin by noting the somewhat counterintuitive fact that a classical algorithm that merely queries $f$ according to samples from the distribution $\mu$ performs no better than an exhaustive search algorithm\footnote{This phenomenon was recently discussed in the somewhat different context of screening for terrorists \cite{press09}.}.

Indeed, consider a classical algorithm that consists of repeatedly obtaining a sample $y$ from $\mu$, then querying $f(y)$. If $x$ is the marked element, the expected number of samples from $\mu$ required until $x$ is found is exactly $1/p_x$ (assuming that $p_x>0$). Thus the expected number of samples used is
\[ \sum_{x=1}^n p_x \left( \frac{1}{p_x} \right) = n; \]
the algorithm might as well have just carried out an exhaustive search to find $x$. Being given access to a {\em quantum} oracle $O_\mu$ producing a coherent superposition corresponding to the distribution $\mu$, however, will turn out to be very useful.


\subsection{Quantum algorithm}

Our quantum algorithm will be based on the {\em amplitude amplification} primitive of Brassard et al \cite{brassard02}. 
For completeness, an explicit definition of amplitude amplification is given below, as Algorithm \ref{alg:wamp}.

\begin{algorithm}
\KwIn{Function $f:[n] \rightarrow \{0,1\}$ such that $f$ takes the value 1 on precisely one input $x$; oracle operator $O_\mu:\ket{0} \mapsto \ket{\mu}$; inverse $O_\mu^{-1}$; positive integer $k$ (number of iterations)}
\KwOut{The marked element $x$, or fail}
create initial state $\ket{\mu} = O_\mu\ket{0}$\;
apply operator $-O_\mu I_{\ket{0}} O_\mu^{-1} I_{\ket{x}}$ $k$ times to $\ket{\mu}$\;
measure in computational basis, obtaining outcome $y$\;
\eIf{f(y)=1}{
  \Return{y}\;
}{
  \Return{\mbox{fail}}\;
}
\caption{Amplitude amplification \cite{brassard02}\vspace{\baselineskip}}
\label{alg:wamp}
\end{algorithm}

The notation $I_{\ket{\psi}}$ denotes reflection about the state $\ket{\psi}$; it is well-known that the operator $I_{\ket{x}}$, where $I_{\ket{x}}\ket{y} = (-1)^{f(y)} \ket{y}$, can be implemented using one query to $f$.
%
%
The following result was shown by Brassard et al in \cite{brassard02} (with somewhat different terminology).

\begin{lem}
Applying Algorithm \ref{alg:wamp} with $k$ iterations returns the location of the marked element with probability $\sin^2((2k+1)\arcsin{\sqrt{p_x}})$, using $k+1$ queries to $O_\mu$, $k$ queries to $O_{\mu}^{-1}$, and $k+1$ queries to $f$.
\end{lem}

We now use Algorithm \ref{alg:wamp} as a subroutine in an algorithm which finds the marked element with certainty, and takes advantage of $O_\mu$ to reduce the expected number of queries used. The algorithm is a modified version of previous ``exponential searching'' algorithms of Brassard et al \cite{brassard02}, and Boyer et al \cite{boyer98}. The main difference is that the algorithm gives up after a certain number of iterations and reverts to the exact variant of standard Grover search \cite{hoyer00,brassard02}. This change can make a significant difference to the overall query complexity. The algorithm is stated as Algorithm \ref{alg:unknown} below.

\begin{algorithm}[htp]
\KwIn{Function $f:[n] \rightarrow \{0,1\}$ such that $f$ takes the value 1 on precisely one input $x$; oracle operator $O_\mu:\ket{0} \mapsto \ket{\mu}$; inverse $O_\mu^{-1}$; real $k>1$}
\KwOut{The marked element $x$}
\For{$j = 0$ to $\lfloor \log_k \sqrt{n} \rfloor$}{
sample from distribution $\mu$\;
\If{marked element found}{
  \Return{marked element}\;
}
pick $i$ uniformly at random from integers $\{0,\dots,\lfloor k^j\rfloor - 1\}$\;
perform $i$ iterations of amplitude amplification\;
\If{marked element found}{
  \Return{marked element}\;
}
}
perform exact Grover search for one marked element on $[n]$\;
\Return{marked element}\;
\caption{Quantum search with unknown probability distribution\vspace{\baselineskip}}
\label{alg:unknown}
\end{algorithm}

It will turn out to be possible to give a close analysis of the expected query complexity of Algorithm \ref{alg:unknown}, including constants (which we will round to integers; these could be optimised further). The analysis follows the approach taken by Boyer et al \cite{boyer98} to bound the performance of their quantum search algorithm for an unknown number of marked elements.

\newcounter{amp}\setcounter{amp}{\value{thm}}
\newcounter{ampsec}\setcounter{ampsec}{\value{section}}
\begin{prop}
\label{prop:amp}
On input $x$, when called with $k \approx 1.162$, Algorithm \ref{alg:unknown} uses an expected number of at most $\min\{83/\sqrt{p_x} + 4/3,53\sqrt{n} \}$ queries to each of $f$, $O_\mu$, $O_\mu^{-1}$.
\end{prop}

\begin{proof}
Deferred to Appendix.
\end{proof}

\begin{cor}
\label{cor:qstar}
Let $\mathcal{A}$ denote Algorithm \ref{alg:unknown}. Then there are constants $K$, $L$, $M$ such that
\[ T_{\mathcal{A}}^*(\mu) \le K \left( \sum_{x,p_x > 1/n} \sqrt{p_x} \right) + L \sqrt{n} \left( \sum_{x,p_x \le 1/n} p_x \right) + M. \]
\end{cor}

\begin{proof}
Apply Algorithm \ref{alg:unknown}, and use Proposition \ref{prop:amp} to calculate the average number of queries used with respect to the distribution $\mu$.
\end{proof}


\subsection{Power law distributions}

As with the case of a known probability distribution, power law distributions provide a natural class of examples for search with an unknown probability distribution. For some of these distributions, Algorithm \ref{alg:unknown} can be used to obtain significant speed-ups over any classical algorithm, even one with complete knowledge of the distribution.

\newcounter{powerlawtwo}\setcounter{powerlawtwo}{\value{thm}}
\newcounter{powerlawtwosec}\setcounter{powerlawtwosec}{\value{section}}
\begin{prop}
\label{prop:powerlaw2}
Let $\mu = (p_x)$ be a probability distribution where $p_x \propto x^k$ for some constant $k<0$, and let $\mathcal{A}$ denote Algorithm \ref{alg:unknown}. Then
\begin{equation*}
T^*_{\mathcal{A}}(\mu) = \left\{ 
\begin{array}{ll}
  O(\sqrt{n}) &\,\,[-1 \le k < 0] \\
  O(n^{-(1/2+1/k)}) &\,\,[-2 < k < -1] \\
  O(\log n) &\,\,[k=-2] \\
  O(1) &\,\,[k<-2]
  \end{array}
\right.
\end{equation*}
\end{prop}

\begin{proof}
Deferred to Appendix.
\end{proof}


\section*{Acknowledgements}

This work was supported by the EC-FP6-STREP network QICS. I would like to thank Aram Harrow for pointing out reference \cite{press09}.


\section*{Appendix}

In this appendix we collect some proofs from throughout the paper.

\setcounter{subsection}{0}
\renewcommand{\thesubsection}{A.\arabic{subsection}}

\subsection{Proofs from Section \ref{sec:known}}

\setcounter{section}{\value{powerlawsec}}
\setcounter{thm}{\value{powerlaw}}
\begin{prop}
Let $\mu = (p_x)$, $x \in [n]$ be a probability distribution where $p_x \propto x^k$ for some constant $k<0$. Then
\begin{equation*}
\begin{array}{cc}
D(\mu) = \left\{ 
\begin{array}{ll}
 \Theta(n) &\,\,[-1 < k < 0] \\
 \Theta(n/\log n) &\,\,[k=-1] \\
 \Theta(n^{k+2}) &\,\,[-2 < k < -1] \\
 \Theta(\log n) &\,\,[k=-2] \\
 \Theta(1) &\,\,[k<-2]
\end{array}
\right.
&, \mbox{ and }
Q(\mu) = \left\{ 
\begin{array}{ll}
  \Theta(\sqrt{n}) &\,\,[-1 < k < 0] \\
  \Theta(\sqrt{n}/\log n) &\,\,[k=-1] \\
  \Theta(n^{k+3/2}) &\,\,[-3/2 < k < -1] \\
  \Theta(\log n) &\,\,[k=-3/2] \\
  \Theta(1) &\,\,[k<-3/2]
  \end{array}
\right.
\end{array}
\end{equation*}
\end{prop}

\begin{proof}
From the statement of the proposition, $p_x = \alpha x^k$ for some $\alpha$. We first estimate the normalising constant $\alpha$. The constraint that $\sum_{x=1}^n p_x = 1$ implies that $1/\alpha = \sum_{x=1}^n x^k$. Estimating this sum by an integral, we have that 
\[ \int_1^n x^k\,dx \le 1/\alpha \le 1 + \int_1^n x^k\,dx, \]
implying that, for $k \neq -1$,
\[ \frac{n^{k+1}-1}{k+1} \le 1/\alpha \le \frac{n^{k+1}-1}{k+1} + 1. \]
By (\ref{eqn:deterministic}), if it also holds that $k \neq -2$,
\[ D(\mu) = \alpha \sum_{x=1}^n x^{k+1} \ge \alpha \int_1^n x^{k+1}\,dx \ge \frac{(n^{k+2}-1)(k+1)}{(n^{k+1}+k)(k+2)}. \]
The upper bound on $D(\mu)$ is very similar. It is easy to see that this proves the deterministic half of the proposition, except for the cases $k=-1$ and $k=-2$, which can be verified directly. In the quantum case, by Proposition \ref{prop:average}, for $k \neq -1,-3/2$,
\[ Q(\mu) \le \alpha \pi e \sum_{x=1}^n x^{k+1/2} \le \alpha \pi e \left(1+ \int_{1}^{n+1} x^{k+1/2}\,dx\right) = \pi e \left(\frac{((n+1)^{k+3/2}+k+1/2)(k+1)}{(n^{k+1}-1)(k+3/2)} \right). \]
Again, the lower bound is similar and the special cases $k = -1$, $k=-3/2$ can be verified directly.
\end{proof}


\subsection{Proofs from Section \ref{sec:unknown}}

In this section, we give the proof of Proposition \ref{prop:amp}, which bounds the performance of Algorithm \ref{alg:unknown}. We will need a lemma of Boyer et al \cite{boyer98}, which we translate into our terminology.

\setcounter{thm}{0}
\renewcommand{\thethm}{A.\arabic{thm}}

\begin{lem}[Boyer et al \cite{boyer98}]
\label{lem:uniformamp}
If an integer $r$ is picked from the range $\{0,\dots,m-1\}$ uniformly at random, and $r$ iterations of amplitude amplification are performed, the probability of finding the marked element is exactly
\[ P_m = \frac{1}{2} - \frac{\sin(4m \arcsin \sqrt{p_x})}{8m\sqrt{p_x(1-p_x)}} .\]
In particular, $P_m \ge 1/4$ whenever
\[ m \ge \frac{1}{2\sqrt{p_x(1-p_x)}}. \]
\end{lem}

We are now ready to prove Proposition \ref{prop:amp}. The proof is similar to a result of Boyer et al \cite{boyer98}, but with somewhat more detail.

\renewcommand{\thethm}{\arabic{section}.\arabic{thm}}

\setcounter{section}{\value{ampsec}}
\setcounter{thm}{\value{amp}}
\begin{prop}
On input $x$, when called with $k \approx 1.162$, Algorithm \ref{alg:unknown} uses an expected number of at most $\min\{83/\sqrt{p_x} + 4/3,53\sqrt{n} \}$ queries to each of $f$, $O_\mu$, $O_\mu^{-1}$.
\end{prop}

\begin{proof}
We upper bound the expected number of queries to $f$ used by Algorithm \ref{alg:unknown}, which implies the same bound on the number of queries to $O_\mu$ and $O_\mu^{-1}$. The bound will be in terms of $k$, and eventually minimised over $k$ such that $1 < k < 4/3$. However, changing $k$ will only change the number of queries used by a constant factor. Let $T_j$ denote the expected number of queries to $f$ used if the marked element is found in the $j$'th iteration of the loop. Then
\[ T_j = \sum_{i=0}^{j-1} \frac{\lfloor k^i \rfloor + 3}{2} \le \frac{k^j}{k-1}, \]
an inequality which holds for $1 < k < 4/3$ and can be proven by induction on $j$.

We first find an upper bound by considering the {\em worst-case} number of queries used. If it has not been found previously, the marked element is guaranteed to be found in the last, exact Grover search step. Thus the number of queries to $f$ used is at most
\[ \sum_{j=1}^{\lfloor \log_k \sqrt{n} \rfloor + 1} T_j + \left\lceil \frac{\pi}{4} \sqrt{n} \right\rceil \le \frac{1}{k-1} \sum_{j=1}^{\lfloor \log_k \sqrt{n} \rfloor + 1} k^j + \left\lceil \frac{\pi}{4} \sqrt{n} \right\rceil \le \left( \frac{k^2}{(k-1)^2} + \frac{\pi}{4} \right) \sqrt{n}. \]
This deals with one half of the statement of the proposition. For the remainder of the proof, we restrict to the case $p_x \ge 1/n$ (as the case $p_x \le 1/n$ will be covered by the above bound), and also assume that $p_x \le 3/4$; this assumption will be removed at the end.

We now assume that the ``sample from distribution $\mu$'' step always fails (this can only increase the number of queries used). The expected number of queries to $f$ used by Algorithm \ref{alg:unknown} is then upper bounded by
\[ \sum_{j=0}^{\lfloor \log_k \sqrt{n} \rfloor} \left( \prod_{i=0}^{j-1} (1-P_{\lfloor k^i\rfloor}) \right) P_{\lfloor k^j\rfloor} T_{j+1} + \left( \prod_{i=0}^{\lfloor \log_k \sqrt{n} \rfloor} (1-P_{\lfloor k^i\rfloor}) \right)\left(T_{\lfloor \log_k \sqrt{n} \rfloor+1} + \left\lceil \frac{\pi}{4} \sqrt{n} \right\rceil \right). \]
To bound this expression, we split the first sum into two parts. First, we have
\[ \sum_{j=0}^{\lfloor \log_k 1/\sqrt{p_x} \rfloor} \left( \prod_{i=0}^{j-1} (1-P_{\lfloor k^i\rfloor}) \right) P_{\lfloor k^j\rfloor} T_{j+1} \le T_{\lfloor \log_k 1/\sqrt{p_x} \rfloor + 1} \le \frac{k}{k-1} \frac{1}{\sqrt{p_x}}. \]
Using Lemma \ref{lem:uniformamp} and the fact that $p_x \le 3/4$, it holds for $j \ge \lfloor \log_k 1/\sqrt{p_x} \rfloor + 1$ that
\[ \prod_{i=0}^j (1-P_{\lfloor k^i\rfloor}) \le \left( \frac{3}{4} \right)^{j-\lfloor \log_k 1/\sqrt{p_x} \rfloor}. \]
Thus
\[
\begin{split}
\sum_{j=\lfloor \log_k 1/\sqrt{p_x} \rfloor+1}^{\lfloor \log_k \sqrt{n} \rfloor} & \left( \prod_{i=0}^{j-1} (1-P_{\lfloor k^i\rfloor}) \right) P_{\lfloor k^j\rfloor} T_{j+1}\\
\le \,&\, \frac{k}{k-1} \sum_{j=\lfloor \log_k 1/\sqrt{p_x} \rfloor + 1}^{\lfloor \log_k \sqrt{n} \rfloor} \left( \frac{3}{4} \right)^{j-\lfloor \log_k 1/\sqrt{p_x} \rfloor - 1} k^j \\
\le \,&\, \frac{k}{k-1} \left( \frac{4}{3} \right)^{\lfloor \log_k 1/\sqrt{p_x} \rfloor + 1} \sum_{j=0}^{\infty} \left( \frac{3k}{4} \right)^{j+\lfloor \log_k 1/\sqrt{p_x} \rfloor + 1}\\
\le \,&\, \frac{4k^2}{(k-1)(4-3k)} \frac{1}{\sqrt{p_x}},
\end{split}
\]
and also
\[
\begin{split}
\left( \prod_{i=0}^{\lfloor \log_k \sqrt{n} \rfloor}(1-P_{\lfloor k^i\rfloor}) \right)& \left(T_{\lfloor \log_k \sqrt{n} \rfloor+1} + \left\lceil \frac{\pi}{4} \sqrt{n} \right\rceil\right)\\
\le\,&\, \left(\frac{3}{4}\right)^{\lfloor \log_k \sqrt{n} \rfloor - \lfloor \log_k 1/\sqrt{p_x} \rfloor} \left(T_{\lfloor \log_k \sqrt{n} \rfloor+1} + \left\lceil \frac{\pi}{4} \sqrt{n} \right\rceil\right)\\
\le\,&\, \left(\frac{3}{4}\right)^{\log_k \sqrt{n} - \log_k 1/\sqrt{p_x} - 1}\left(\left(\frac{k}{k-1} + \frac{\pi}{4}\right) \sqrt{n} + 1\right)\\
\le\,&\, \left( \frac{4k}{3(k-1)} + \frac{\pi}{3} \right)\frac{1}{\sqrt{p_x}} + \frac{4}{3},
\end{split}
\]
where we use again the restriction that $p_x \ge 1/n$. Combining these bounds gives the following overall upper bound on the expected number of queries used:
\[ \left( \frac{k}{k-1} + \frac{4k^2}{(k-1)(4-3k)} + \frac{4k}{3(k-1)} + \frac{\pi}{3} \right) \frac{1}{\sqrt{p_x}} + \frac{4}{3}. \]
Minimising the bracketed expression over $k$ using simple calculus gives that the minimum is found at $k \approx 1.162$; for this value of $k$, we obtain a bound on the expected number of queries used that is approximately
\[ \frac{82.646}{\sqrt{p_x}} + \frac{4}{3}. \]
Finally, consider the case that $p_x \ge 3/4$. In this case, one can find a bound by assuming that only the sampling step in each iteration of the loop can succeed, and ignoring the amplitude amplification step. Using this assumption, the number of queries to $f$ used is upper bounded by
\[ \frac{3}{4} \sum_{j=0}^{\lfloor \log_k \sqrt{n} \rfloor} \left(\frac{1}{4}\right)^j (1 + T_j) + \left(\frac{1}{4}\right)^{\lfloor \log_k \sqrt{n}\rfloor + 1} \left(T_{\lfloor \log_k \sqrt{n} \rfloor + 1} + \left\lceil \frac{\pi}{4} \sqrt{n} \right\rceil \right), \]
which is readily seen to be upper bounded by
\[ \frac{3}{(k-1)(4-k)} + \frac{k}{k-1} + 2 + \frac{\pi}{4}. \]
Inserting the previously found value of $k$, $k \approx 1.162$, gives an upper bound of an expected $\approx 16.500$ queries used in this case, and completes the proof.
\end{proof}


\setcounter{section}{\value{powerlawtwosec}}
\setcounter{thm}{\value{powerlawtwo}}
\begin{prop}
Let $\mu = (p_x)$ be a probability distribution where $p_x \propto x^k$ for some constant $k<0$, and let $\mathcal{A}$ denote Algorithm \ref{alg:unknown}. Then
\begin{equation*}
T^*_{\mathcal{A}}(\mu) = \left\{ 
\begin{array}{ll}
  O(\sqrt{n}) &\,\,[-1 \le k < 0] \\
  O(n^{-(1/2+1/k)}) &\,\,[-2 < k < -1] \\
  O(\log n) &\,\,[k=-2] \\
  O(1) &\,\,[k<-2]
  \end{array}
\right.
\end{equation*}
\end{prop}

\begin{proof}
As in the proof of Proposition \ref{prop:powerlaw}, $p_x = \alpha x^k$ for some $\alpha$, and for $k \neq -1$,
\[ \alpha \le \frac{k+1}{n^{k+1}-1}. \]
Define $x_0 = \max \{x \in [n]: p_x \ge 1/n \} = \lfloor(\alpha n)^{-1/k}\rfloor$. Then, by Corollary \ref{cor:qstar},
\[ Q^*(\mu) \le K \sqrt{\alpha} \sum_{x=1}^{x_0} x^{k/2} + L\,\alpha \sqrt{n} \sum_{x=x_0+1}^{n} x^k + M, \]
for some constants $K$, $L$, $M$, implying that for $k \neq -1$, $k \neq -2$,
\beas Q^*(\mu) &\le& K \sqrt{\alpha} \left(1 + \int_1^{x_0} x^{k/2}\,dx\right) + L\,\alpha \sqrt{n} \int_{x_0}^{n} x^k\,dx + M\\
&\le& K \sqrt{\alpha}\left(1 - \frac{2}{k+2} + \frac{2 x_0^{k/2+1}}{k+2} \right) + L \frac{\alpha\sqrt{n}}{k + 1}(n^{k+1} - x_0^{k+1}) + M\\
&\le& K \sqrt{\frac{k+1}{n^{k+1}-1}} \left(1 - \frac{2}{k+2} + \frac{2 x_0^{k/2+1}}{k+2} \right) + L \frac{\sqrt{n}(n^{k+1}-x_0^{k+1})}{n^{k+1}-1} + M.
\eeas
Now note that
\[ x_0 \approx \left(\frac{n(k+1)}{n^{k+1}-1}\right)^{-1/k} \]
is $\Theta(n)$ for $k>-1$, and $\Theta(n^{-1/k})$ for $k < -1$. Inserting this into the previous expression we obtain the claimed results for the cases $k\neq -1$, $k \neq -2$. These remaining special cases can be verified directly.
\end{proof}



\begin{thebibliography}{10}

\bibitem{ambainis01}
A.~Ambainis and R.~de~Wolf.
\newblock Average-case quantum query complexity.
\newblock {\em J. Phys. A: Math. Gen.}, 34:6741--–6754, 2001.
\newblock \url{quant-ph/9904079}.

\bibitem{beals01}
R.~Beals, H.~Buhrman, R.~Cleve, M.~Mosca, and R.~de~Wolf.
\newblock Quantum lower bounds by polynomials.
\newblock {\em J. ACM}, 48(4):778--797, 2001.
\newblock \url{quant-ph/9802049}.

\bibitem{boyer98}
M.~Boyer, G.~Brassard, P.~H\o yer, and A.~Tapp.
\newblock Tight bounds on quantum searching.
\newblock {\em Fortschr. Phys.}, 46(4--5):493--505, 1998.
\newblock \url{quant-ph/9605034}.

\bibitem{brassard02}
G.~Brassard, P.~H{\o }yer, M.~Mosca, and A.~Tapp.
\newblock Quantum amplitude amplification and estimation.
\newblock {\em Quantum Computation and Quantum Information: A Millennium
  Volume}, pages 53--74, 2002.
\newblock \url{quant-ph/0005055}.

\bibitem{buhrman02}
H.~Buhrman and R.~de~Wolf.
\newblock Complexity measures and decision tree complexity: a survey.
\newblock {\em Theoretical Computer Science}, 288:21--43, 2002.

\bibitem{grover97}
L.~Grover.
\newblock Quantum mechanics helps in searching for a needle in a haystack.
\newblock {\em Phys. Rev. Lett.}, 79(2):325--328, 1997.
\newblock \url{quant-ph/9706033}.

\bibitem{grover02}
L.~Grover and T.~Rudolph.
\newblock Creating superpositions that correspond to efficiently integrable
  probability distributions, 2002.
\newblock \url{quant-ph/0208112}.

\bibitem{hoyer00}
P.~H{\o }yer.
\newblock Arbitrary phases in quantum amplitude amplification.
\newblock {\em Phys. Rev. A.}, 62:052304, 2000.
\newblock \url{quant-ph/0006031}.

\bibitem{motwani95}
R.~Motwani and P.~Raghavan.
\newblock {\em Randomized algorithms}.
\newblock Cambridge University Press, 1995.

\bibitem{nielsen00}
M.~A. Nielsen and I.~L. Chuang.
\newblock {\em Quantum computation and quantum information}.
\newblock Cambridge University Press, 2000.

\bibitem{press09}
W.~H. Press.
\newblock Strong profiling is not mathematically optimal for discovering rare
  malfeasors.
\newblock {\em Proceedings of the National Academy of Sciences},
  106(6):1716--1719, 2009.

\bibitem{zalka99}
C.~Zalka.
\newblock Grover's quantum searching algorithm is optimal.
\newblock {\em Phys. Rev. A.}, 60(4):2746--2751, 1999.
\newblock \url{quant-ph/9711070}.

\end{thebibliography}

\end{document}